\documentclass[twocolumn]{autart}    

\usepackage{graphicx}          

\usepackage{amsmath}
\usepackage{amssymb}
\usepackage{lipsum}
\usepackage{color}
\usepackage{savesym}
\savesymbol{AND}
\savesymbol{OR}
\savesymbol{NOT}
\savesymbol{TO}
\savesymbol{COMMENT}
\savesymbol{BODY}
\savesymbol{IF}
\savesymbol{ELSE}
\savesymbol{ELSIF}
\savesymbol{FOR}
\savesymbol{WHILE}

\usepackage{algorithmic}

\newtheorem{problem}{Problem}                               
\newtheorem{definition}{Definition}
\newtheorem{proposition}{Proposition}
\newtheorem{proof}{Proof}
\newtheorem{remark}{Remark}

\graphicspath{{figures/}}

\begin{document}

\begin{frontmatter}
\runtitle{Control of Nonlinear Switched Systems}  

\title{Control of Nonlinear Switched Systems Based on Validated Simulation} 

\thanks[footnoteinfo]{This paper was not presented at any IFAC 
meeting. A short version of this paper appeared in SNR'16 \cite{NL_minimator}.
Corresponding author A. Le Co\"ent Tel. +33147407429.}

\author[cmla]{Adrien Le Co\"ent$^\star$}\ead{adrien.le-coent@ens-cachan.fr},    
\author[ensta]{Julien Alexandre dit Sandretto}\ead{alexandre@ensta.fr},
\author[ensta]{Alexandre Chapoutot}\ead{chapoutot@ensta.fr},               
\author[lsv]{Laurent Fribourg}\ead{fribourg@lsv.ens-cachan.fr}  

\address[cmla]{CMLA, ENS Cachan, CNRS, Universit\'e Paris-Saclay, 61 av. du Pr\'esident Wilson, 94235 Cachan Cedex, France}  
\address[ensta]{U2IS, ENSTA ParisTech, Universit\'e Paris-Saclay, 828, Boulevard des Mar\'echaux, 91762 Palaiseau Cedex, France}             
\address[lsv]{LSV, ENS Cachan, CNRS, Universit\'e Paris-Saclay, 61 av. du Pr\'esident Wilson, 94235 Cachan Cedex, France}        

\begin{keyword}                           
Nonlinear control systems, reachability, formal methods, numerical simulation, control system synthesis             
\end{keyword}                             

\begin{abstract}                          
  We present an algorithm of control synthesis for nonlinear switched
  systems, based on an existing procedure of state-space bisection and
  made available for nonlinear systems with the help of validated
  simulation. The use of validated simulation also permits to take
  bounded perturbations and varying parameters into account. 
  It is particularly interesting for safety critical applications,
  such as in aeronautical, military or medical fields. 
  The whole approach is entirely guaranteed and the induced controllers are
  \textit{correct-by-design}. 
\end{abstract}

\end{frontmatter}

\section{Introduction}


We focus here on switched control systems, a class of hybrid systems
recently used with success in various domains such as automotive
industry and power electronics. These systems are merely described by
piecewise dynamics, periodically sampled with a given period. At each
period, the system is in one and only one mode, decided by a control
rule \cite{fribourg2014finite,liberzon2012switching}. Moreover, the considered 
systems can switch between any two modes instantaneously. This simplification can be easily 
by-passed by the addition of intermediate facticious modes.

In this paper, we consider that these modes are represented by
nonlinear ODEs.  In order to compute the control of a switched system,
we do need the solution of differential equations. In the general
case, differential equations can not be integrated formally, and a
numerical integration scheme is used to approximate the state of the
system. With the objective of computing a guaranteed control, we base
our approach on validated simulation (also called ``reachability
analysis'').  The \emph{guaranteed} or \emph{validated} solution of
ODEs using interval arithmetic is mainly based on two kinds of methods
based on: i) Taylor series
\cite{Moore66,Nedialkov,LiSt07,Dzetkulic:2015fk} ii) Runge-Kutta
schemes \cite{BM06,Gajda:2008fk,BCD13,report}. The former is the
oldest method used in interval analysis community because the
expression of the bound of a Taylor series is simple to obtain.
Nevertheless, the family of Runge-Kutta methods is very important in
the field of numerical analysis.  Indeed, Runge-Kutta methods have
several interesting stability properties which make them suitable for
an important class of problems. Our tool \cite{dynibex} implements
Runge-Kutta based methods which prove their efficiency at low order
for short simulation (fixed by sampling period of controller).

In the methods of symbolic analysis and control of hybrid systems, the
way of representing sets of state values and computing reachable sets
for systems defined by autonomous ordinary differential equations
(ODEs), is fundamental (see, e.g.,
\cite{girard2005reachability,althoff2013reachability}).  Many tools
using, eg. linearization or hybridization of these dynamics are now
available (e.g., Spacex \cite{frehse2011spaceex}, Flow*
\cite{chen2013flow}, iSAT-ODE~\cite{eggers2008sat}).
An interesting approach appeared recently, based on the propagation of
reachable sets using guaranteed Runge-Kutta methods with adaptive step
size control (see \cite{BCD13,immler2015verified}).  An originality of
the present work is to use such guaranteed integration methods in the
framework of switched systems. This notion of guarantee of the results is very 
interesting, because we are mainly interested into critical domain, such as aeronautical, military and medical ones. 
Other symbolic approaches for control synthesis of
switched systems include the construction of a discrete abstraction 
of the original system
on a grid of the state space. This can be done
by computing symbolic models that are approximately bisimilar \cite{girard2010approximately} or 
approximately alternatingly similar \cite{zamani2012symbolic} 
to the original system.
Another recent symbolic approach relies on
feedback refinement relations \cite{reissig2015feedback}. We compare our work with the 
last two approaches, which are the closest related methods since the associated 
tools (respectively PESSOA \cite{Mazo2010} and SCOTS \cite{SCOTS}) are used to perform 
control synthesis on switched systems without any stability 
assumptions, such as the present method.


The paper is divided as follows. In Section~\ref{sec:switched}, we
introduce some preliminaries on switched systems and some notation
used in the following. In Section~\ref{sec:simulation}, the guaranteed
integration of nonlinear ODEs is presented. In
Section~\ref{sec:minimator}, we present the main algorithm of
state-space bisection used for control synthesis.  In
Section~\ref{sec:experimentations}, the whole approach is tested on
three examples of the literature.
We give some performance tests and compare 
our approach with the state-of-the-art tools in section \ref{sec:comparison}.
We conclude in section \ref{sec:conclu}.

\section{Switched systems}
\label{sec:switched}

Let us consider the nonlinear switched system
\begin{equation}
 \dot x(t) = f_{\sigma (t)}(x(t),d(t))
 \label{eq:sys}
\end{equation}
defined for all $t \geq 0$, where $x(t) \in \mathbb{R}^n$ is the state
of the system, $\sigma(\cdot) : \mathbb{R}^+ \longrightarrow U$ is the
switching rule, and $d(t) \in \mathbb{R}^m$ is a bounded
perturbation. The finite set $U = \{ 1, \dots , N \}$ is the set of
switching modes of the system.  We focus on sampled switched systems:
given a sampling period $\tau >0$, switchings will occur at times
$\tau$, $2\tau$, \dots{} The switching rule $\sigma(\cdot)$ is thus
piecewise constant, we will consider that $\sigma(\cdot)$ is constant
on the time interval $\lbrack (k-1) \tau , k \tau )$ for $k \geq 1$.
We call ``\emph{pattern}'' a finite sequence of modes $\pi =
(i_1,i_2,\dots,i_k) \in U^k$.  With such a control input, and under a
given perturbation $d$, we will denote by $\mathbf{x}(t; t_0,
x_0,d,\pi)$ the solution at time $t$ of the system
\begin{equation}
\begin{aligned}
  \dot x(t) & =  f_{\sigma (t)}(x(t),d(t)), \\
  x(t_0) & =  x_0, \\
  \forall j \in \{1,\dots,k\}, \ \sigma(t) & =  i_j \in U \ \text{for} \ t
  \in \lbrack (j-1) \tau , j \tau ).
\end{aligned}
 \label{eq:sampled-sys}
\end{equation}

We address the problem of synthesizing a
  state-dependent switching rule $\tilde \sigma(x)$
  for~\eqref{eq:sampled-sys} in order to verify some properties. The
  problem is formalized as follows:
\begin{problem}[Control Synthesis Problem]
 Let us consider a sampled switched system~\eqref{eq:sampled-sys}.
 Given three sets $R$, $S$, and $B$, with $R \cup B
   \subset S$ and $R \cap B = \emptyset$, find a rule $\tilde
 \sigma(x)$ such that, for any $x(0)\in R$
 \begin{itemize}
 \item \textit{$\tau$-stability}\footnote{This definition of stability
   is different from the stability in the Lyapunov sense.}: $x(t)$
   returns in $R$ infinitely often, at some multiples of sampling time
   $\tau$.
  \item \textit{safety}: $x(t)$ always stays in $S \backslash B$.
 \end{itemize}
 \label{prob:nl_control}
\end{problem}

 Under the above-mentioned notation, we propose a
 procedure which solves this problem by constructing a law $\tilde
 \sigma(x)$, such that for all $x_0 \in R$, and under the unknown
 bounded perturbation $d$, there exists $\pi = \tilde \sigma(x_0) \in
 U^k$ for some $k$ such that:
 \begin{equation*}
   \left\{
     \begin{aligned}
       \mathbf{x}(t_0 + k\tau; t_0, x_0,d,\pi) \in R
       \\
       \forall t \in [t_0,t_0 + k\tau], \quad
       \mathbf{x}(t; t_0, x_0,d,\pi) \in S
       \\
       \forall t \in [t_0,t_0 + k\tau], \quad
       \mathbf{x}(t; t_0, x_0,d,\pi) \notin B
\end{aligned}
\right.
 \end{equation*}


 Such a law permits to perform an infinite-time state-dependent
 control. The synthesis algorithm is described in
 Section~\ref{sec:minimator} and involves guaranteed set based
 integration presented in the next section, the main underlying tool
 is interval analysis \cite{Moore66}.  To tackle this problem, we
 introduce some definitions. In the following, we will often use the
 notation $\lbrack x \rbrack \in \mathbb{IR}$ (the set of intervals
 with real bounds) where $\lbrack x \rbrack = \lbrack\underline{x},
 \overline{x}\rbrack=\{ x \in \Rset \mid \underline{x} \leqslant x
 \leqslant \overline{x} \}$ denotes an interval. By an abuse of
 notation $[x]$ will also denote a vector of intervals, \emph{i.e.}, a
 Cartesian product of intervals, a.k.a. a \emph{box}. In the
 following, the sets $R$, $S$ and $B$ are given under the form of
 boxes.

 \begin{definition}[Initial Value Problem (IVP)]
   Consider an ODE with a given
   initial condition
   \begin{equation}
     \label{eq:ivp}
     \dot{x}(t) = f(t, x(t), d(t))\quad\text{with}
     \quad x(0) \in X_0, \ d(t) \in \lbrack d \rbrack,
   \end{equation}
   with $f:\Rset^+\times\Rset^n\times \Rset^m\rightarrow\Rset^n$
   assumed to be continuous in $t$ and $d$ and globally Lipschitz in
   $x$. We assume that parameters $d$ are bounded (used to represent a
   perturbation, a modeling error, an uncertainty on
   measurement,~\dots).  An \emph{IVP} consists in finding a function
   $x(t)$ described by the ODE~\eqref{eq:ivp} for all $d(t)$ lying in
   $\lbrack d \rbrack$ and for all the initial conditions in $X_0$.
\end{definition}

 \begin{definition}
   Let $X \subset \mathbb{R}^n$ be a box of the state space. Let $\pi
   = (i_1,i_2,\dots,i_k) \in U^k$. The \emph{successor set} of $X$
   via $\pi$, denoted by $Post_{\pi}(X)$, is the (over-approximation
   of the) image of  $X$ induced by
   application of the pattern $\pi$, \emph{i.e.}, the solution at time
   $t = k \tau$ of
  \begin{equation}
    \label{eq:ivp_post}
    \begin{aligned}
      \dot x(t) &=  f_{\sigma (t)}(x(t),d(t)), \\
      x(0) & =  x_0 \in X,   \\
      \forall t \geq 0, & \quad d(t) \in \lbrack d \rbrack, \\
       \forall j \in \{1,\dots,k\},& \quad \sigma(t) = i_j \in U \ \text{for}
       \ t \in \lbrack (j-1) \tau , j \tau ).
    \end{aligned}
 \end{equation}
  \label{def:post}
 \end{definition}

 \begin{definition}
   Let $X \subset \mathbb{R}^n$ be a box of the state space.  Let $\pi
   = (i_1,i_2,\dots,i_k) \in U^k$.  We denote by $Tube_{\pi}(X)$ the
   union of boxes covering the trajectories of
   IVP~\eqref{eq:ivp_post}, which construction is detailed in
   Section~\ref{sec:simulation}.
  \label{def:tube}
 \end{definition}

\section{Validated simulation}
\label{sec:simulation}



In this section, we describe our approach for validated simulation
based on Runge-Kutta methods \cite{BCD13,report}.

A numerical integration method computes a sequence of approximations
$(t_n, x_n)$ of the solution $x(t;x_0)$ of the IVP defined in
Equation~\eqref{eq:ivp} such that $x_n \approx x(t_n;x_{n-1})$. The
simplest method is Euler's method in which $t_{n+1}=t_n+h$ for some
step-size $h$ and $x_{n+1}=x_n+h\times f(t_n,x_n, d)$; so the
derivative of $x$ at time $t_n$, $f(t_n,x_n, d)$, is used as an
approximation of the derivative on the whole time interval to perform
a linear interpolation. This method is very simple and fast, but
requires small step-sizes. More advanced methods coming from the
Runge-Kutta family use a few intermediate computations to improve the
approximation of the derivative. The general form of an explicit
$s$-stage Runge-Kutta formula, that is using $s$ evaluations of $f$,
is
\begin{equation}
  \begin{aligned}
    x_{n+1}  = x_n + h \sum_{i=1}^s b_i k_i\enspace, \\
    k_1  = f\big(t_n,\, x_n, d\big)\enspace,\\
    k_i  = f \Big(t_n + c_i h,\, x_n + h \sum_{j=1}^{i-1} a_{ij}k_j, d\Big)
    , \ i = 2,3,\dots,s\enspace.
    \label{eq:ki}
  \end{aligned}
\end{equation}
The coefficients $c_i$, $a_{ij}$ and $b_i$ fully characterize the
method. To make Runge-Kutta validated, the challenging question is how
to compute a bound on the distance between the true solution and the
numerical solution, defined by $x(t_n;x_{n-1}) - x_n$. This distance
is associated to the \emph{local truncation error} (LTE) of the
numerical method.

To bound the LTE, we rely on \textit{order condition}~\cite{HNW93}
respected by all Runge-Kutta methods. This condition states that a
method of this family is of order $p$ iff the $p+1$ first coefficients
of the Taylor expansion of the solution and the Taylor expansion of
the numerical methods are equal. In consequence, LTE is proportional
to the Lagrange remainders of Taylor expansions.  Formally, LTE is
defined by (see \cite{BCD13}):
\begin{multline}
  \label{eq:truncation-error}
  x(t_n;x_{n-1}) - x_n = \\
  \hspace{5mm}\frac{h^{p+1}}{(p+1)!} \left( f^{(p)}\left(\xi,x(\xi; x_{n-1}), d
  \right) - \frac{d^{p+1}\phi}{dt^{p+1}}(\eta) \right)
  \\
  \xi\in]t_n, t_{n+1}[ \text{ and } \eta\in]t_n, t_{n+1}[\enspace.
\end{multline}
The function $f^{(n)}$ stands for the $n$-th derivative of function
$f$ w.r.t. time $t$ that is $\frac{d^n f}{dt^n}$ and $h=t_{n+1}-t_n$
is the step-size.  The function $\phi:\Rset\to\Rset^n$ is defined by
$\phi(t)= x_n + h \sum_{i=1}^s b_i k_i(t)$ where $k_i(t)$ are defined
as Equation~\eqref{eq:ki}.

The challenge to make Runge-Kutta integration schemes safe w.r.t. the
true solution of IVP is then to compute a bound of the result of
Equation~\eqref{eq:truncation-error}. In other words we have to bound
the value of $f^{(p)}\left(\xi, x(\xi;x_{n-1}), d\right)$ and the
value of $\frac{d^{p+1}\phi}{dt^{p+1}}(\eta)$. The latter expression
is straightforward to bound because the function $\phi$ only depends
on the value of the step-size $h$, and so does its $(p+1)$-th
derivative. The bound is then obtain using the affine arithmetic \cite{AffineA97,alexandre2016validated}. 

However, the expression $f^{(p)}\left(\xi, x(\xi;x_{n-1}), d\right)$
is not so easy to bound as it requires to evaluate $f$ for a
particular value of the IVP solution $x(\xi;x_{n-1})$ at an unknown
time $\xi \in ]t_n, t_{n+1}[$. The solution used is the same as the
one found in~\cite{Nedialkov,BM06} and it requires to bound the
solution of IVP on the interval $[t_n, t_{n+1}]$. This bound is
usually computed using the Banach's fixpoint theorem applied with the
Picard-Lindel\"of operator, see \cite{Nedialkov}. This operator is
used to compute an enclosure of the solution $[\tilde{x}]$ of IVP over
a time interval $[t_n, t_{n+1}]$, that is for all $t \in [t_n,
t_{n+1}]$, $x(t; x_{n-1}) \in [\tilde{x}]$. We can hence bound
$f^{(p)}$ substituting $x(\xi;x_{n-1})$ by $[\tilde{x}]$.

For a given pattern of switched modes $\pi = (i_1,\dots,i_k)
\in U^k$ of length $k$, we are able to compute, for $j \in
\{1,..,k\}$, the enclosures:
\begin{itemize}
 \item $[x_j] \ni x(t_j)$;
 \item $[\tilde{x}_j] \ni x(t), \ \text{for} \ t \in \lbrack
   (j-1)\tau,j\tau\rbrack$.
\end{itemize}
with respect to the system of IVPs:

\begin{equation}
  \left\{
  \begin{array}{c}
    \dot x(t) = f_{\sigma (t)}(t,x(t),d(t)),\\
    \nonumber x(t_0=0) \in [x_0] , d(t) \in [d],\\
    \sigma(t) = i_1, \forall t \in [0,t_1], t_1=\tau\\
    \vdots\\
    \dot x(t) = f_{\sigma (t)}(t,x(t),d(t)),\\
    \nonumber x(t_{k-1}) \in [x_{k-1}], d(t) \in [d],\\
    \sigma(t) = i_k, \forall t \in [t_{k-1},t_k], t_k=k\tau
  \end{array}
  \right.
\end{equation}

Thereby, the enclosure $Post_{\pi}(\lbrack x_0 \rbrack)$ is included
in $[x_k]$ and $Tube_{\pi}(\lbrack x_0 \rbrack)$ is included in
$\bigcup_{j=1,..,k} [\tilde{x}_j]$. This applies for all initial
states in $\lbrack x_0 \rbrack$ and all disturbances $d(t) \in [d]$. A
view of enclosures computed by the validated simulation for one
solution obtained for Example~\ref{ex2} is shown in
Figure~\ref{fig:post_tube}.

\begin{figure}[ht]
 \centering
 \includegraphics[trim = 4cm 3cm 4cm 4cm, clip, width=0.45\textwidth]{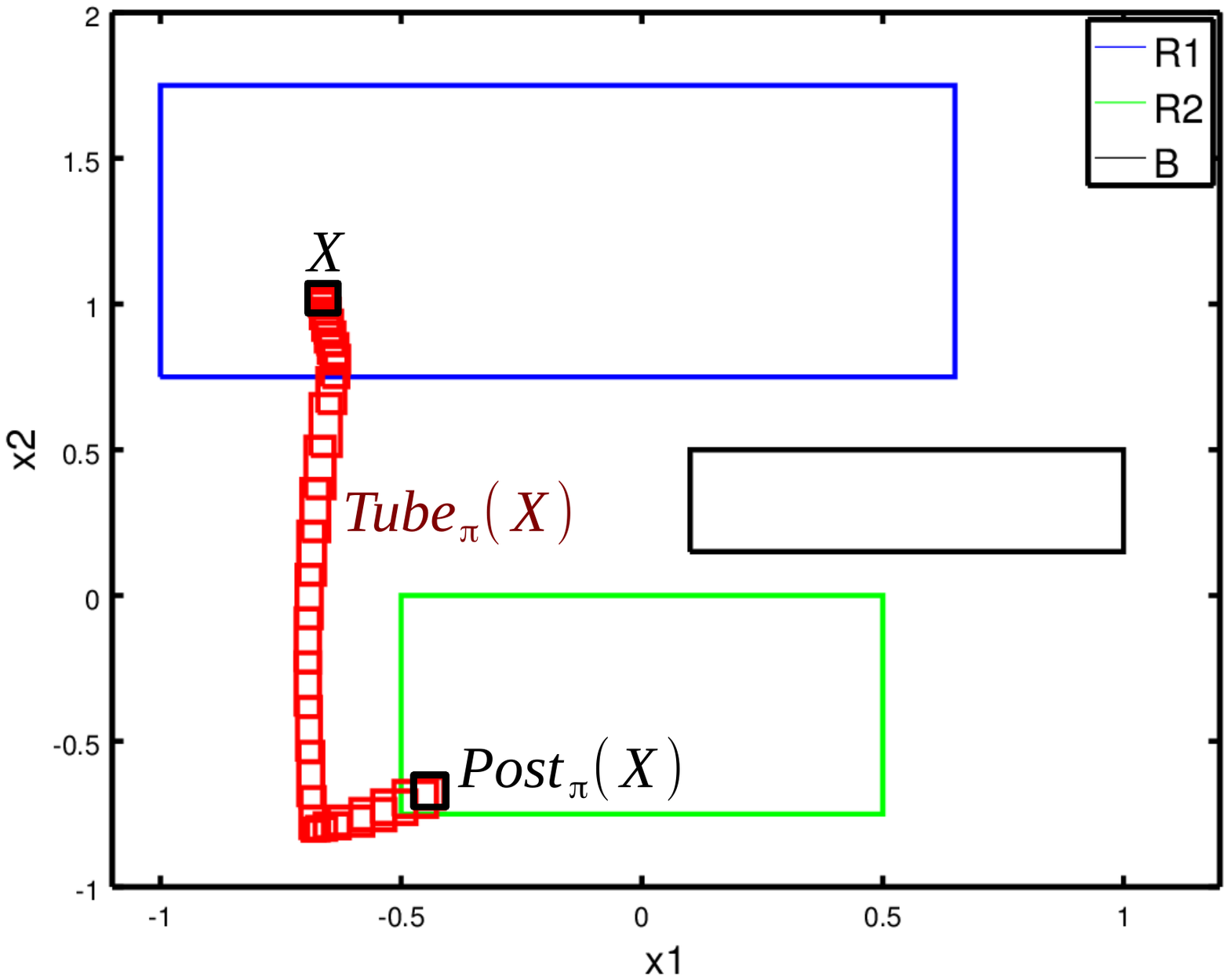}
 \caption{Functions $Post_{\pi}(X)$ and $Tube_{\pi}(X)$ for the initial box $X=[-0.69,-0.64] \times [1,1.06]$, 
with a pattern $\pi = (1,3,0)$.}
 \label{fig:post_tube}
\end{figure}

\section{The state-space bisection algorithm}
\label{sec:minimator}

\subsection{Principle of the algorithm}

We describe here the algorithm solving the control synthesis problem (see Problem \ref{prob:nl_control},
Section \ref{sec:switched}). Given the
input boxes $R$, $S$, $B$, and given two positive integers $K$ and
$D$, the algorithm provides, when it succeeds, a decomposition
$\Delta$ of $R$ of the form $\{ V_i, \pi_i \}_{i \in I}$, with the
properties:

$\bigcup_{i \in I} V_i = R$,

$\forall i \in I, \ Post_{\pi_i}(V_i) \subseteq R$,

$\forall i \in I, \ Tube_{\pi_i}(V_i) \subseteq S$,

$\forall i \in I, \ Tube_{\pi_i}(V_i) \bigcap B = \emptyset$.

The sub-boxes $\{ V_i \}_{i \in I}$ are obtained by repeated
bisection.  At first, function $Decomposition$ calls sub-function
$Find\_Pattern$ which looks for a pattern $\pi$ of length at most $K$
such that $Post_{\pi}(R) \subseteq R$, $Tube_{\pi}(R) \subseteq S$ and
$Tube_{\pi}(R) \bigcap B = \emptyset$.  If such a pattern $\pi$ is
found, then a uniform control over $R$ is found (see
Figure~\ref{fig:scheme}(a)). Otherwise, $R$ is divided into two
sub-boxes $V_1$, $V_{2}$, by bisecting $R$ w.r.t. its longest
dimension. Patterns are then searched to control these sub-boxes (see
Figure~\ref{fig:scheme}(b)). If for each $V_i$, function
$Find\_Pattern$ manages to get a pattern $\pi_i$ of length at most $K$
verifying $Post_{\pi_i}(V_i) \subseteq R$, $Tube_{\pi_i}(V_i)
\subseteq S$ and $Tube_{\pi_i}(V_i) \bigcap B = \emptyset$, then it is
done.  If, for some $V_j$, no such pattern is found, the procedure is
recursively applied to $V_j$.  It ends with success when every sub-box
of $R$ has a pattern verifying the latter conditions, or fails when
the maximal degree of decomposition $D$ is reached.  The algorithmic
form of functions $Decomposition$ and $Find\_Pattern$ is given in Figures \ref{fig:decomposition} (cf. \ref{fig:findpattern} and in
\cite{fribourg2014finite,ulrich} for the linear case).

\begin{figure}[ht]
 \centering
 \includegraphics[trim = 2cm 6cm 4cm 5.5cm, clip, width=0.45\textwidth]{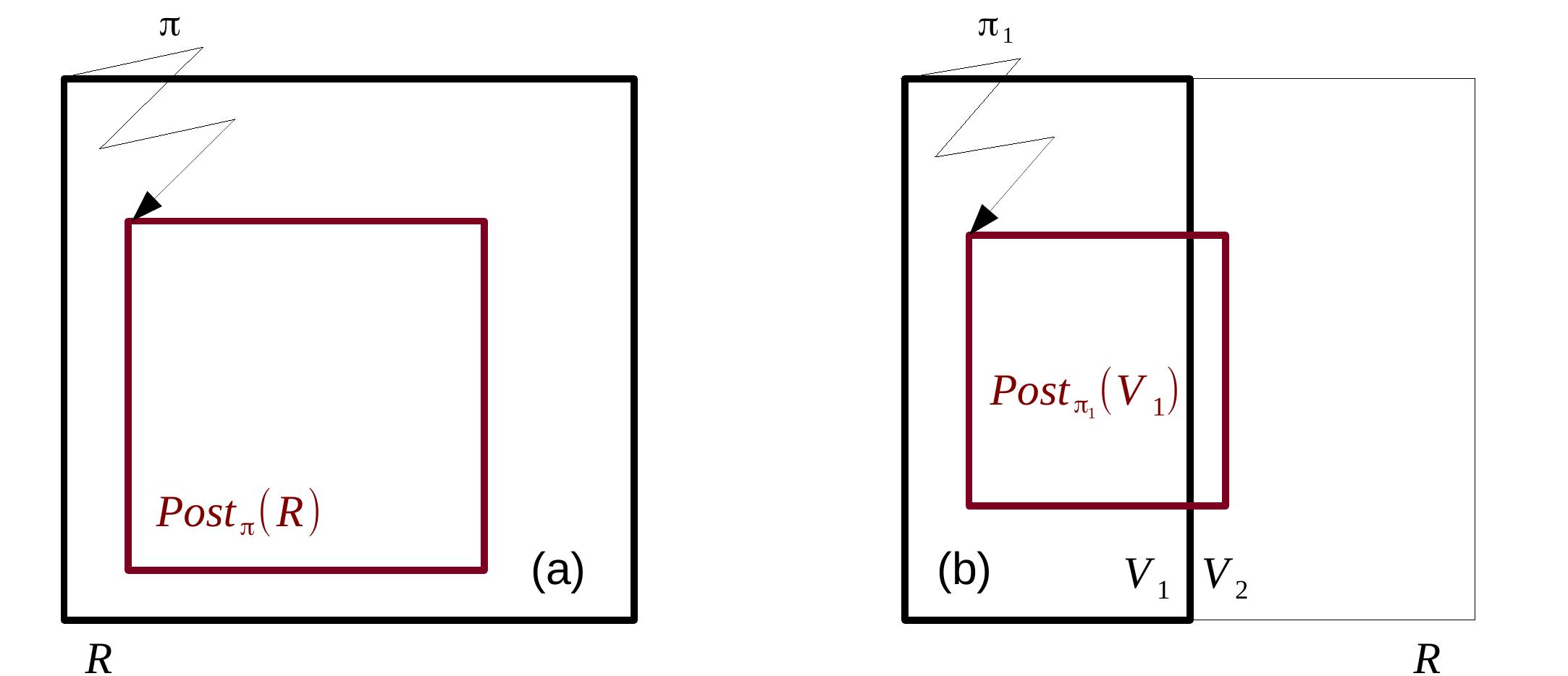}
 \caption{Principle of the bisection method.}
 \label{fig:scheme}
\end{figure}

Having defined the control synthesis method, we now introduce the main
result of this paper, stated as follows:
 \begin{proposition}
  The algorithm of Figure \ref{fig:decomposition} with input $(R,R,S,B,D,K)$
  outputs, when it successfully terminates,
  a decomposition $\{ V_i,\pi_i \}_{i \in I}$ of~$R$ which solves Problem~\ref{prob:nl_control}. 
 \end{proposition}

 \begin{proof}
 Let $x_0 = x(t_0=0)$ be an initial condition belonging to~$R$. If the
 decomposition has terminated successfully, we have $\bigcup_{i \in I}
 V_i = R$, and $x_0$ thus belongs to $V_{i_0}$ for some $i_0\in I$.
 We can thus apply the pattern $\pi_{i_0}$ associated to
 $V_{i_0}$. Let us denote by $k_0$ the length of $\pi_{i_0}$. We have:
 \begin{itemize}
  \item $\mathbf{x}(k_0\tau;0,x_0,d,\pi_{i_0}) \in R$
  \item $\forall t \in [0, k_0\tau], \quad
    \mathbf{x}(t;0,x_0,d,\pi_{i_0}) \in S$
  \item $\forall t \in [0, k_0\tau], \quad
    \mathbf{x}(t;0,x_0,d,\pi_{i_0}) \notin B$
 \end{itemize}
Let $x_1 = \mathbf{x}(k_0\tau;0,x_0,d,\pi_{i_0}) \in R$ be the
state reached after application of $\pi_{i_0}$ and let $t_1 = k_0
\tau$.  State $x_1$ belongs to $R$, it thus belongs to $V_{i_1}$ for
some $i_1 \in I$, and we can apply the associated pattern $\pi_{i_1}$
of length $k_1$, leading to:
 \begin{itemize}
  \item $\mathbf{x}(t_1 + k_1\tau;t_1,x_1,d,\pi_{i_1}) \in R$
  \item $\forall t \in [t_1, t_1 + k_1\tau], \quad
    \mathbf{x}(t;t_1,x_1,d,\pi_{i_1}) \in S$
  \item $\forall t \in [t_1, t_1 + k_1\tau], \quad
    \mathbf{x}(t;t_1,x_1,d,\pi_{i_1}) \notin B$
 \end{itemize}
 We can then iterate this procedure from the new state $x_2 =
 \mathbf{x}(t_1 + k_1\tau;t_1,x_1,d,\pi_{i_1}) \in R$.  This can be
 repeated infinitely, yielding a sequence of points belonging to $R$
 $x_0,x_1,x_2,\dots$ attained at times $t_0,t_1,t_2,\dots$, at which
 the patterns $\pi_{i_0},\pi_{i_1},\pi_{i_2},\dots$ are applied.

 We furthermore have that all the trajectories stay in $S$ and never
 cross $B$: $ \forall t \in \mathbb{R}^+, \exists k \geq 0, \ t \in
 \lbrack t_k , t_{k+1} \rbrack$ and $ \forall t \in \lbrack t_k ,
 t_{k+1} \rbrack,\ \mathbf{x} ( t; t_k, x_k, d,
 \pi_{i_k}) \in S,\ \mathbf{x} (t;t_k, x_k, d,
 \pi_{i_k}) \notin B $.  The trajectories thus return infinitely often
 in $R$, while always staying in $S$ and never crossing $B$.
 \end{proof}

 \begin{remark}
   Note that it is possible to perform reachability from a set $R_1$
   to another set $R_2$ by computing $Decomposition(R_1,R_2,S,B,D,K)$.
   The set $R_1$ is thus decomposed with the objective to send its
   sub-boxes into $R_2$, i.e. for a sub-box $V$ of $R_1$, patterns
   $\pi$ are searched with the objective $Post_{\pi}(V) \subseteq R_2$
   (see Example~\ref{ex2}).
 \end{remark}

  \begin{figure}
  \fbox{
\begin{minipage}{0.42\textwidth}
\begin{algorithmic}
\STATE{\textbf{Function:} $Decomposition(W,R,S,B,D,K)$}
  \STATE{\begin{center}\line(1,0){150}\end{center}}
\STATE{\quad \textbf{Input:} A box $W$, a box $R$, a box $S$, a box $B$, a degree $D$ of bisection,
a length $K$ of input pattern}\STATE{\quad \textbf{Output:}$\langle\{(V_i,\pi_i)\}_{i},True\rangle$ or $\langle\_ ,False\rangle$}
  \STATE{\begin{center}\line(1,0){150}\end{center}}
  \STATE{ $(\pi,b) := Find\_Pattern(W,R,S,B,K)$}
  \IF{$b=True$}{
    \STATE{\textbf{return} $\langle\{(W,Pat)\},True\rangle$}
  }
  \ELSE
    \IF{$D = 0$} \RETURN{$\langle \_,False\rangle$} \ELSE
  \STATE{Divide equally $W$ into $(W_1, W_{2})$ \FOR{$i=1,2$}\STATE{\small{$(\Delta_i,b_i)$ := $Decomposition(W_i,R,S,B,D - 1,K)$}}\ENDFOR
  \RETURN $(\bigcup_{i=1,2} \Delta_i,\bigwedge_{i=1,2} b_i)$ } \ENDIF
 \ENDIF
\end{algorithmic}
\end{minipage}
}
\caption{Algorithmic form of Function $Decomposition$.}
\label{fig:decomposition}
\end{figure}

\begin{figure}
\fbox{
\begin{minipage}{0.42\textwidth}
 \begin{algorithmic}
 \STATE{\textbf{Function:} $Find\_Pattern(W,R,S,B,K)$}
   \STATE{\begin{center}\line(1,0){150}\end{center}}
 \STATE{\quad \textbf{Input:}A box $W$, a box $R$, a box $S$, a box $B$, a length $K$ of input pattern}
 \STATE{\quad \textbf{Output:}$\langle \pi,True\rangle$ or $\langle\_, False\rangle$}
   \STATE{\begin{center}\line(1,0){150}\end{center}}
   \FOR{$i=1\dots K$} \STATE{$\Pi :=$ set of input patterns of length $i$}
   \WHILE{$\Pi$ is non empty} \STATE{Select $\pi$ in $\Pi$}
   \STATE{$\Pi:= \Pi\setminus  \{\pi\}$}
   \IF{$Post_{\pi}(W) \subseteq R$ \AND $Tube_{\pi}(W) \subseteq S$ \AND $Tube_{\pi}(W) \bigcap B = \emptyset$ }{\RETURN{$\langle \pi,True\rangle$}} \ENDIF
   \ENDWHILE
   \ENDFOR
   \RETURN{$\langle \_,False \rangle$}

 \end{algorithmic}
\end{minipage}
}
\caption{Algorithmic form of Function $Find\_Pattern$.}
\label{fig:findpattern}
\end{figure}

 \subsection{The research of patterns}

 We propose here an improvement of the function $Find\_Pattern$
 given in \cite{NL_minimator,fribourg2014finite,ulrich}, which is a naive 
 testing of all the patterns of growing length (up to $K$).
 
 The improved function, denoted here by $Find\_Pattern2$, 
 exploits heuristics to prune the search tree of
 patterns. The algorithmic form of $Find\_Pattern2$ is given in Figure \ref{fig:findpattern2}.
 It relies on a new data structure consisting of a list 
 of triplets containing:
 \begin{itemize}
  \item An initial box $V \subset \mathbb{R}^n$,
  \item A {\em current} box $Post_{\pi}(V)$, image of $V$ by the pattern $\pi$,
  \item The associated pattern $\pi$.
 \end{itemize}
 For any element $e$ of a list of this type, we denote by $e.Y_{init}$ the initial box,
 $e.Y_{current}$ the {\em current} box, and by $e.\Pi$ the associated pattern.
 We denote by $e_{current} = takeHead(\mathcal{L})$ the 
 element on top of a list $\mathcal L$ (this element is removed from list $\mathcal L$).
 The function $putTail(\cdot,\mathcal{L})$ adds an element at the end of the list $\mathcal L$.

 Let us suppose one wants to control a box $X \subseteq R$.
 The list $\mathcal{L}$ of Figure \ref{fig:findpattern2} 
 is used to store the intermediate computations leading to possible solutions (patterns
 sending $X$ in $R$ while never crossing $B$ or  $\mathbb{R}^n \setminus S$). It is initialized
 as $\mathcal{L} = \{ \left(X,X, \emptyset \right) \}$. 
 First, a testing of all the control modes is performed (a set simulation starting from
 $W$ during time $\tau$ is 
  computed for all the modes in $U$). 
The first level of branches is thus tested exhaustively. If a branch leads to crossing $B$
or $\mathbb{R}^n \setminus S$, the branch is cut. Otherwise, either a solution is found 
or an intermediate state is added to $\mathcal{L}$. The next level of
branches (patterns of length $2$) is then explored from branches that are not cut. And so
on iteratively. At the end, either the tree is explored up to level $K$ (avoiding the cut branches),
or all the branches have been cut at lower levels. 
List $\mathcal{L}$ is thus of the form  $\{ (X,Post_{\pi_i}(X),\pi_i) \}_{i \in {I_X}}$,
 where for each $i \in {I_X}$ we have $Post_{\pi_i}(X) \subseteq S$ and 
 $Tube_{\pi_i}(X) \bigcap B = \emptyset$. Here, $I_X$ is the set of indexes associated to 
 the stored intermediate solutions, $\vert I_X \vert$ is thus the number 
 of stored intermediate solutions for the initial box $X$.
 The number of stored intermediate solutions grows as the search tree of patterns 
 is explored, then decreases as solutions are validated, branches are cut, or the maximal level $K$
 is reached.
 

 The storage of the intermediate solutions $Post_{\pi_i}(X)$ allows to reuse 
 the computations already performed. Even if the search tree of patterns is visited exhaustively,
 it already allows to obtain much better computation times than with Function $Find\_Pattern$.

  A second list, denoted by $Solution$ in Figure \ref{fig:findpattern2},
 is used to store the validated patterns associated to $X$,
 i.e. a list of patterns of the form $\{ \pi_j \}_{j \in I_X'}$,
 where for each $j \in I_X'$ we have $Post_{\pi_j}(X) \subseteq R$, 
 $Tube_{\pi_j}(X) \bigcap B = \emptyset$ and $Tube_{\pi_j}(X) \subseteq S$. Here, $I_X'$ is 
 the set of indexes associated the the stored validated solutions, $\vert I_X' \vert$ is thus 
 the number of stored validated solutions for the initial box $X$. 
 The number of stored validated solutions can only increase, and we hope that 
 at least one solution is found, otherwise, the initial box $X$ is split in two sub-boxes.

 Note that several solutions can be returned by $Find\_Pattern2$, further optimizations 
could thus be performed, such as returning the pattern minimizing a given cost function.
In practice, and in the examples given below, we return the first validated pattern 
and stop the computation as soon as it is obtained 
(see commented line in Figure \ref{fig:findpattern2}).
 

%
 %

Compared to \cite{fribourg2014finite,ulrich}, this new function highly improves the computation
times, even though the complexity of the two functions is theoretically the same, at most in $O(N^K)$.
A comparison between functions $Find\_Pattern$ and $Find\_Pattern2$ is given in 
Section \ref{sec:comparison}.

 \begin{figure*}[t]
\fbox{
\begin{minipage}{0.9\textwidth}
 \begin{algorithmic}
 \STATE{\textbf{Function:} $Find\_Pattern2(W,R,S,B,K)$}
 \STATE{\begin{center}\line(1,0){150}\end{center}}
 \STATE{\quad \textbf{Input:}A box $W$, a box $R$, a box $S$, a box $B$, a length $K$ of input pattern}
 \STATE{\quad \textbf{Output:}$\langle \pi,True\rangle$ or $\langle\_, False\rangle$}
 \STATE{\begin{center}\line(1,0){150}\end{center}}

 \STATE{$Solution = \{  \emptyset \}$}
 \STATE{$\mathcal{L} = \{ \left(W,W, \emptyset \right) \}$}
 \WHILE{$\mathcal{L} \neq \emptyset$} \STATE{$e_{current}$ = takeHead($\mathcal{L}$)}
 
 \FOR{$i \in U$} 
    \IF{$Post_{i}(e_{current}.Y_{current}) \subseteq R$ \AND $Tube_{i}(e_{current}.{Y_{current}}) \bigcap B = \emptyset$ \AND $Tube_{i}(e_{current}.Y_{current}) \subseteq S$} \STATE{$\text{putTail}(Solution,e_{current}.\Pi + i)$ \quad\quad {\color{blue} /*can be replaced by: ``{\bf return} $\langle e_{current}.\Pi + i,True \rangle$'' */ } }
  \ELSE{ \IF{$Tube_{i}(e_{current}.{Y_{current}}) \bigcap B \neq \emptyset$ \OR
      $Tube_{i}(e_{current}.{Y_{current}}) \nsubseteq S$} \STATE{ discard $e_{current}$ } \ENDIF}
    \ELSE{ \IF{$Tube_{i}(e_{current}.{Y_{current}}) \bigcap B = \emptyset$ \AND $Tube_{i}(e_{current}.Y_{current}) \subseteq S$} 
      \STATE{\IF{$\text{Length}(\Pi)+1 < K$} \STATE{$\text{putTail}(\mathcal{L},\left(e_{\text{current}}.Y_{\text{init}}, Post_i(e_{current}.Y_{current}),e_{\text{current}}.\Pi + i \right))$ }   \ENDIF} \ENDIF}

    \ENDIF
  \ENDFOR
 \ENDWHILE

 \RETURN{$\langle \_,False \rangle$ if no solution is found, or $\langle \pi,True\rangle$, $\pi$ being
 any pattern validated in $Solution$.}

 \end{algorithmic}
\end{minipage}
}

\caption{Algorithmic form of Function $Find\_Pattern2$.}
\label{fig:findpattern2}
\end{figure*}

\section{Experimentations}
\label{sec:experimentations}

In this section, we apply our approach to different case studies taken
from the literature.    Our solver prototype is written in C++ and based on DynIBEX
 \cite{dynibex}.
 The computations times given in the following
have been performed on a 2.80 GHz Intel Core i7-4810MQ CPU with 8 GB
of memory. Note that our algorithm is mono-threaded so all the
experimentation only uses one core to perform the computations.
The results given in this section have been obtained with Function $Find\_Pattern2$.

\subsection{A linear example: boost DC-DC converter}

This linear example is taken from \cite{beccuti2005optimal} and has
already been treated with the state-space bisection method in a linear
framework in \cite{fribourg2014finite}.

The system is a boost DC-DC converter with one switching cell.  There
are two switching modes depending on the position of the switching
cell. The dynamics is given by the equation $\dot x (t) =
A_{\sigma(t)} x(t) + B_{\sigma(t)}$ with $\sigma(t) \in U = \{ 1,2
\}$. The two modes are given by the matrices:

$$ A_1 = \left( \begin{matrix}
          - \frac{r_l}{x_l} & 0 \\ 0 & - \frac{1}{x_c} \frac{1}{r_0 + r_c}
         \end{matrix} \right)  \quad B_1 = \left( \begin{matrix}
         \frac{v_s}{x_l} \\ 0 \end{matrix} \right) $$

$$ A_2 = \left( \begin{matrix} - \frac{1}{x_l} (r_l +
  \frac{r_0.r_c}{r_0 + r_c}) & - \frac{1}{x_l} \frac{r_0}{r_0 + r_c}
  \\ \frac{1}{x_c}\frac{r_0}{r_0 + r_c} & - \frac{1}{x_c}
  \frac{r_0}{r_0 + r_c}
         \end{matrix} \right)  \quad B_2 = \left( \begin{matrix}
         \frac{v_s}{x_l} \\ 0 \end{matrix} \right)  $$

with $x_c = 70$, $x_l = 3$, $r_c = 0.005$, $r_l = 0.05$, $r_0 = 1$,
$v_s = 1$.  The sampling period is $\tau = 0.5$.  The parameters are
exact and there is no perturbation.  We want the state to return
infinitely often to the region $R$, set here to $\lbrack 1.55 , 2.15
\rbrack \times \lbrack 1.0 , 1.4 \rbrack$, while never going out of
the safety set $S = \lbrack 1.54 , 2.16 \rbrack \times \lbrack 0.99 ,
1.41 \rbrack$.

The decomposition was obtained in less than one second with a maximum
length of pattern set to $K = 6$ and a maximum bisection depth of $D =
3$.  A simulation is given in Figure~\ref{fig:NL_0}.

\begin{figure}[!ht]
 \centering
 \includegraphics[scale=0.4]{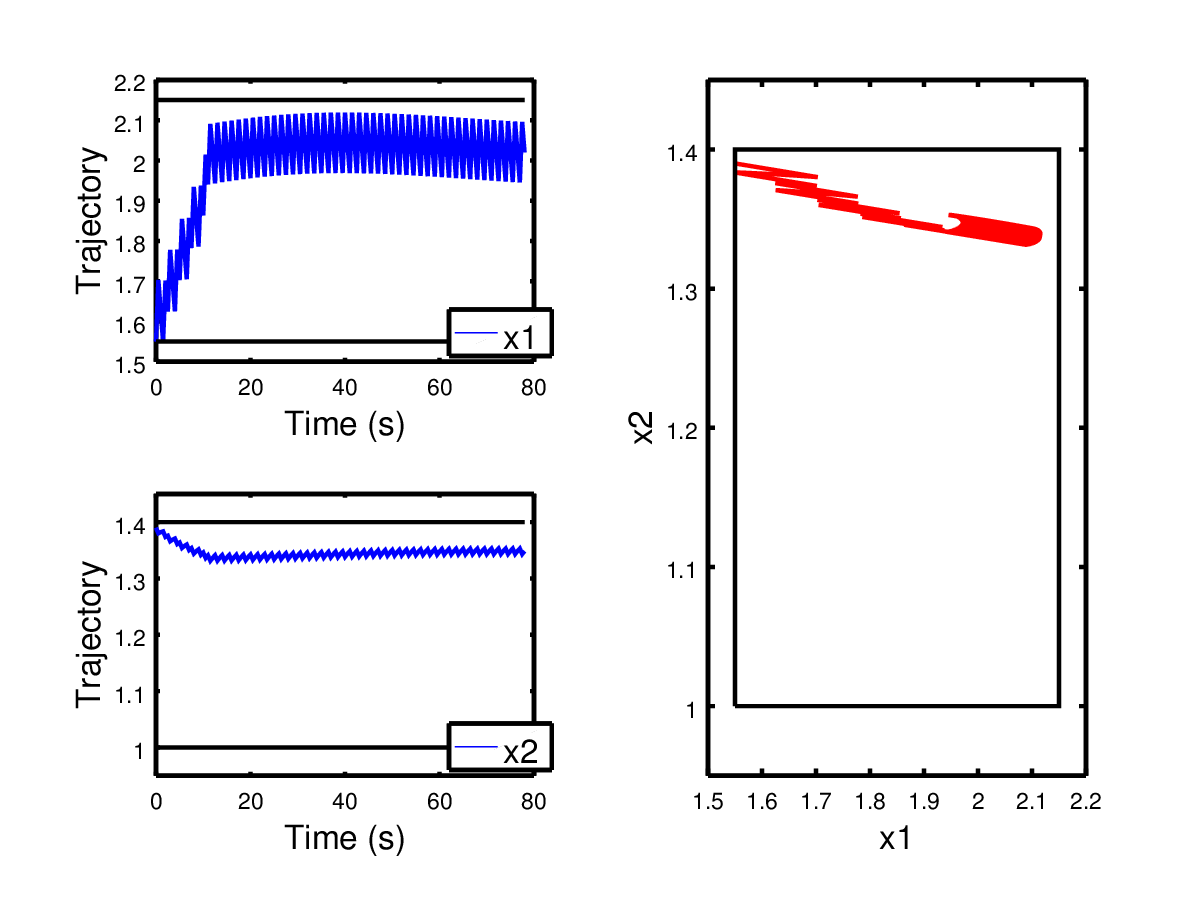}
 \caption{Simulation from the initial condition $(1.55,1.4)$. The box
   $R$ is in plain black. The trajectory is plotted within time for
   the two state variables on the left, and in the state-space plane
   on the right.}
  \label{fig:NL_0}
\end{figure}

\subsection{A polynomial example}
\label{ex2}

We consider the polynomial system taken from \cite{liu2013synthesis}:

\begin{equation}
 \left \lbrack \begin{matrix}
  \dot x_1 \\ \dot x_2
 \end{matrix} \right \rbrack  =
 \left \lbrack \begin{matrix} -x_2 - 1.5 x_1 - 0.5 x_1^3 + u_1 + d_1 \\ x_1 + u_2 + d_2
   \end{matrix} \right \rbrack.
\end{equation}

The control inputs are given by $u = (u_1,u_2) =
K_{\sigma(t)}(x_1,x_2)$, $\sigma(t) \in U = \{ 1,2,3,4 \}$, which correspond to
four different state feedback controllers $K_1(x) = (0,-x_2^2 + 2)$,
$K_2(x) = (0,-x_2)$, $K_3(x) = (2,10)$, $K_4(x) = (-1.5,10)$.  We thus
have four switching modes. The disturbance $d = (d_1,d_2)$ lies in
$\lbrack -0.005,0.005 \rbrack \times \lbrack -0.005,0.005 \rbrack$.
The objective is to visit infinitely often two zones $R_1$ and $R_2$,
without going out of a safety zone $S$, and while never crossing a
forbidden zone $B$.  Two decompositions are performed:
\begin{itemize}
 \item a decomposition of $R_1$ which returns $\{ (V_i,\pi_i) \}_{i
   \in I_1}$ with:

 $\bigcup_{i \in I_1} V_i = R_1$,

 $\forall i \in I_1, \ Post_{\pi_i}(V_i) \subseteq R_2$,

 $\forall i \in I_1, \ Tube_{\pi_i}(V_i) \subseteq S$,

 $\forall i \in I_1, \ Tube_{\pi_i}(V_i) \bigcap B = \emptyset$.

 \item a decomposition of $R_2$ which returns $\{ (V_i,\pi_i) \}_{i
   \in I_2}$ with:

 $\bigcup_{i \in I_2} V_i = R_2$,

 $\forall i \in I_2, \ Post_{\pi_i}(V_i) \subseteq R_1$,

 $\forall i \in I_2, \ Tube_{\pi_i}(V_i) \subseteq S$,

 $\forall i \in I_2, \ Tube_{\pi_i}(V_i) \bigcap B = \emptyset$.

\end{itemize}

The input boxes are the following:

$R_1 = \lbrack -0.5 , 0.5 \rbrack \times \lbrack -0.75 , 0.0 \rbrack$,

$R_2 = \lbrack  -1.0 , 0.65 \rbrack \times \lbrack 0.75 , 1.75 \rbrack$,

$S = \lbrack -2.0 , 2.0 \rbrack \times \lbrack -1.5 , 3.0 \rbrack$,

$B = \lbrack 0.1 , 1.0 \rbrack \times \lbrack 0.15 , 0.5 \rbrack$.

The sampling period is set to $\tau = 0.15$. The decompositions were
obtained in $2$ minutes and $30$ seconds with a maximum length of
pattern set to $K = 12$ and a maximum bisection depth of $D = 5$.  A
simulation is given in Figure~\ref{fig:NL_1} in which the disturbance
$d$ is chosen randomly in $\lbrack -0.005,0.005 \rbrack \times \lbrack
-0.005,0.005 \rbrack$ at every time step.

\begin{figure}[!ht]
 \centering
 \includegraphics[scale=0.4]{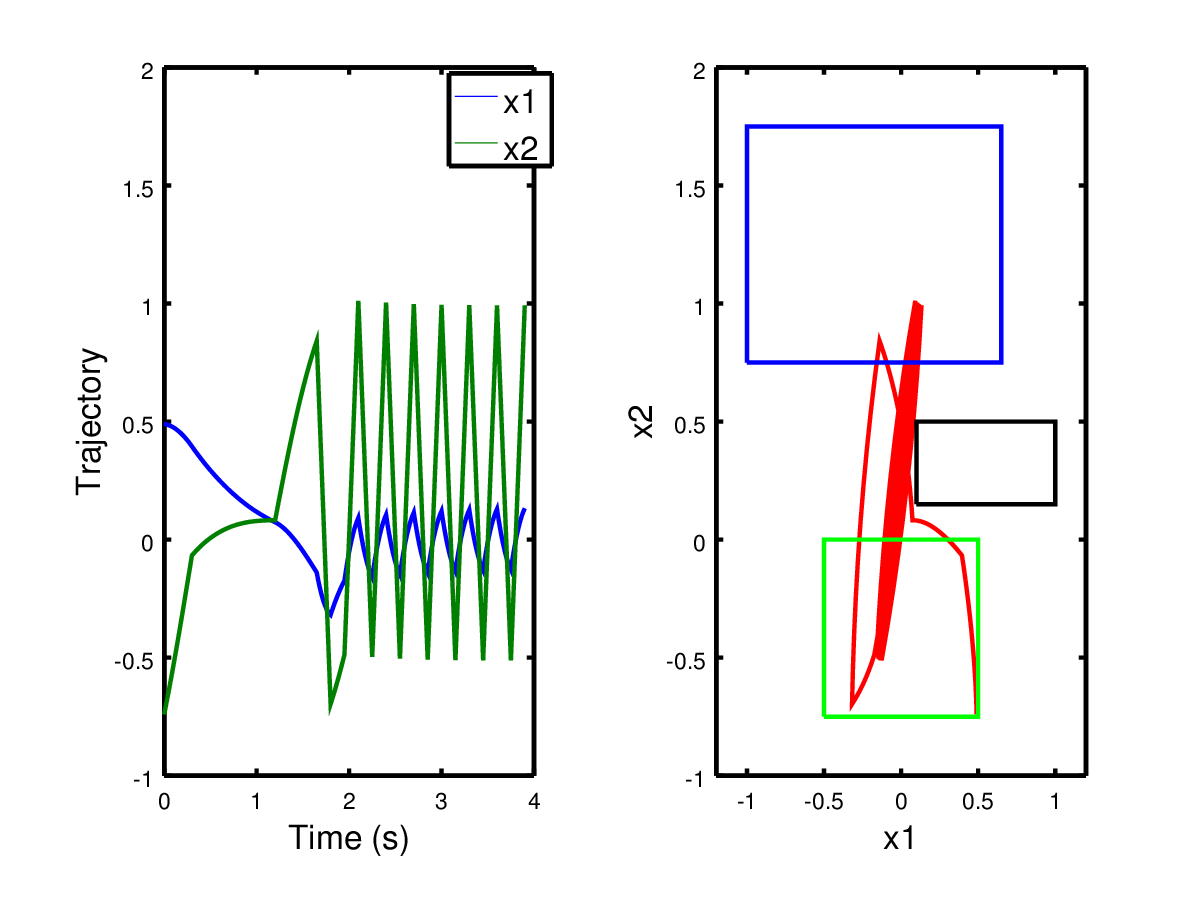}
 \caption{Simulation from the initial condition $(0.5,-0.75)$. The
   trajectory is plotted within time on the left, and in the state
   space plane on the right.  In the sate space plane, the set $R_1$
   is in plain green, $R_2$ in plain blue, and $B$ in plain black.}
 \label{fig:NL_1}
\end{figure}

\subsection{Building ventilation}

We consider a building ventilation application adapted from
\cite{meyer:tel-01232640}.  The system is a four room apartment
subject to heat transfer between the rooms, with the external
environment, with the underfloor, and with human beings.  The dynamics
of the system is given by the following equation:
\begin{multline}
 \frac{d T_i}{dt} = \sum_{j \in \mathcal{N}^\text{*}} a_{ij} (T_j -
 T_i) + \delta_{s_i} b_i (T_{s_i}^4 - T_i ^4 ) \\ + c_i
 \max\left(0,\frac{V_i - V_i^\text{*}}{\bar{ V_i} -
   V_i^{\text{*}}}\right)(T_u - T_i).
\end{multline}

The state of the system is given by the temperatures in the rooms
$T_i$, for $i \in \mathcal{N} = \{ 1 , \dots , 4 \}$.  Room $i$ is
subject to heat exchange with different entities stated by the indexes
$\mathcal{N}^\text{*} = \{1,2,3,4,u,o,c \}$.

The heat
transfer between the rooms is given by the coefficients $a_{ij}$ for
$i,j \in  \mathcal{N}^2$, and the different perturbations are the following:
\begin{itemize}
 \item The external environment: it has an effect on room $i$ with the
   coefficient $a_{io}$ and the outside temperature $T_o$, varying
   between $27^\circ C$ and $30^\circ C$.
  \item The heat transfer through the ceiling: it has an effect on
    room $i$ with the coefficient $a_{ic}$ and the ceiling temperature
    $T_c$, varying between $27^\circ C$ and $30^\circ C$.
  \item The heat transfer with the underfloor: it is given by the
    coefficient $a_{iu}$ and the underfloor temperature $T_u$, set to
    $17^\circ C$ ($T_u$ is constant, regulated by a PID controller).
  \item The perturbation induced by the presence of humans: it is
    given in room $i$ by the term $\delta_{s_i} b_i (T_{s_i}^4 - T_i
    ^4 )$, the parameter $\delta_{s_i}$ is equal to $1$ when someone
    is present in room $i$, $0$ otherwise, and $T_{s_i}$ is a given
    identified parameter.
\end{itemize}

The control $V_i$, $i \in \mathcal{N}$, is applied through the term
$c_i \max(0,\frac{V_i - V_i^\text{*}}{\bar{ V_i} -
  V_i^{\text{*}}})(T_u - T_i)$.  A voltage $V_i$ is applied to force
ventilation from the underfloor to room $i$, and the command of an
underfloor fan is subject to a dry friction.  Because we work in a
switched control framework, $V_i$ can take only discrete values, which
removes the problem of dealing with a ``max'' function in interval
analysis. In the experiment, $V_1$ and $V_4$ can take the values $0$V
or $3.5$V, and $V_2$ and $V_3$ can take the values $0$V or $3$V. This
leads to a system of the form~\eqref{eq:sys} with $\sigma(t) \in U =\{
1, \dots, 16 \}$, the $16$ switching modes corresponding to the
different possible combinations of voltages $V_i$.  The sampling
period is $\tau = 10$s.

The parameters $T_{s_i}$, $V_i^\text{*}$, $\bar V_i$, $a_{ij}$, $b_i$,
$c_i$ are given in \cite{meyer:tel-01232640} and have been identified
with a proper identification procedure detailed in
\cite{meyer2014ecc}.  Note that here we have neglected the term $\sum_{j
  \in \mathcal{N}} \delta_{d_{ij}}c_{i,j} \ast h(T_j - T_i)$ of
\cite{meyer:tel-01232640}, representing the perturbation induced by
the open or closed state of the doors between the rooms. Taking a
``max'' function into account with interval analysis is actually still
a difficult task. However, this term could have been taken into
account with a proper regularization (smoothing).

The decomposition was obtained in $4$ minutes with a maximum length of
pattern set to $K = 2$ and a maximum bisection depth of $D = 4$.  The
perturbation due to human beings has been taken into account by
setting the parameters $\delta_{s_i}$ equal to the whole interval
$\lbrack 0,1 \rbrack$ for the decomposition, and the imposed
perturbation for the simulation is given
Figure~\ref{fig:NL_2_perturbation}.  The temperatures $T_o$ and $T_c$
have been set to the interval $\lbrack27,30\rbrack$ for the
decomposition, and are set to $30^\circ C$ for the simulation.  A
simulation of the controller obtained with the state-space bisection
procedure is given in Figure~\ref{fig:NL_2}, where the control
objective is to stabilize the temperature in $\lbrack 20 , 22 \rbrack
^4$ while never going out of $\lbrack 19 , 23 \rbrack ^4$.

\begin{figure}[ht]
 \centering
 \includegraphics[scale=0.4]{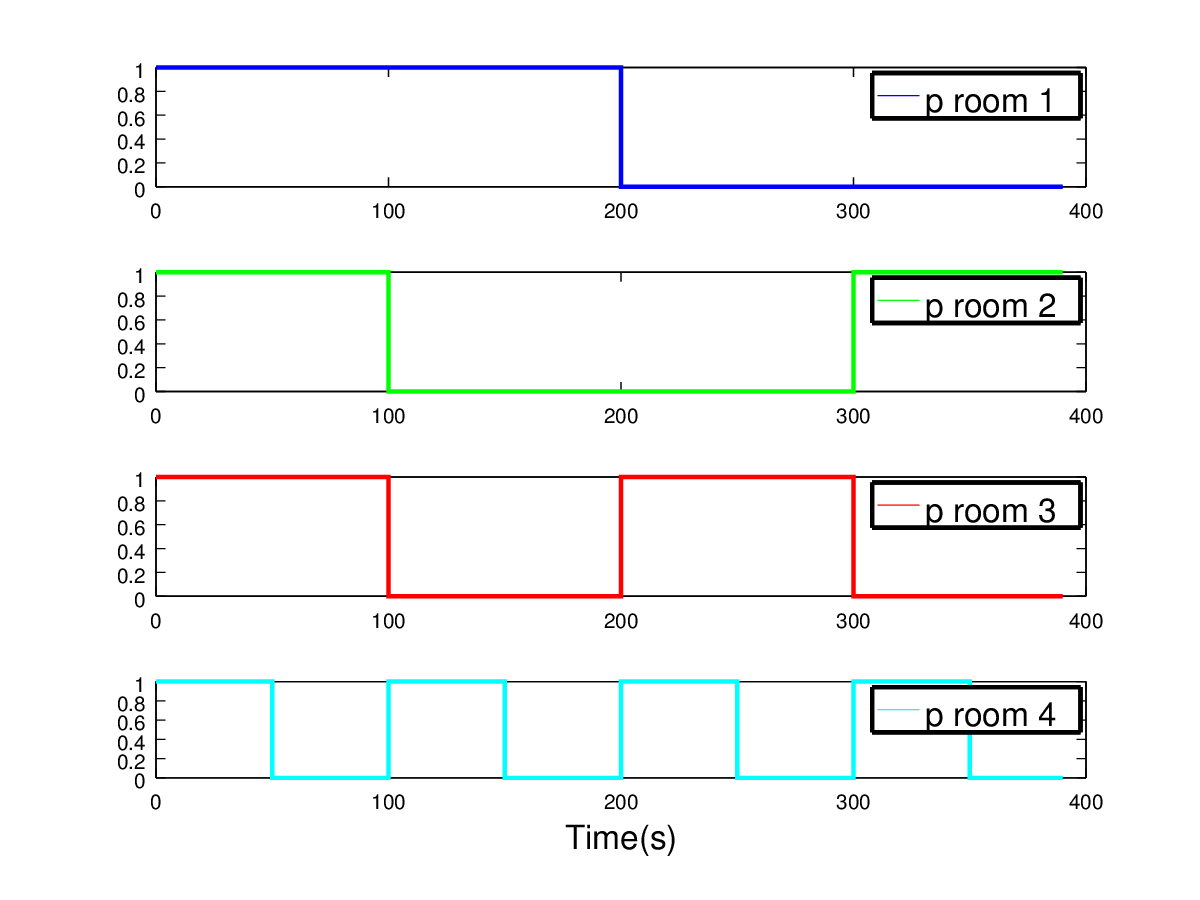}
 \caption{Perturbation (presence of humans) imposed within time in the
   different rooms.}
  \label{fig:NL_2_perturbation}
\end{figure}

\begin{figure}[ht]
 \centering
 \includegraphics[scale=0.35]{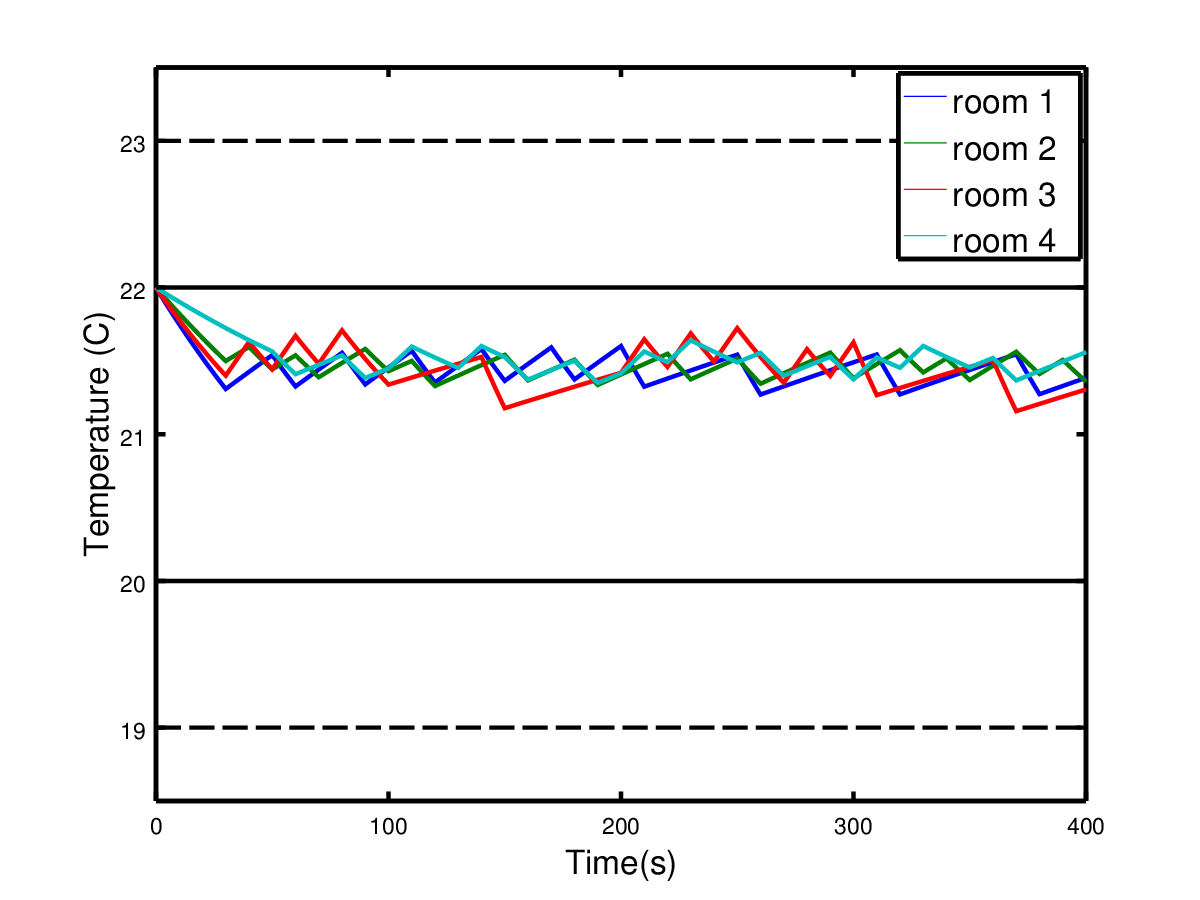}
 \caption{Simulation from the initial condition $(22,22,22,22)$. The
   objective set $R$ is in plain black and the safety set $S$ is in
   dotted black.}
  \label{fig:NL_2}
\end{figure}

\section{Performance tests}
\label{sec:comparison}

We present a comparison of the computation times obtained 
with functions $Find\_Pattern$, $Find\_Pattern2$, and with the state-of-the-art tools
PESSOA \cite{Mazo2010} and SCOTS \cite{SCOTS}.

\begin{table}
\caption{Comparison of $Find\_Pattern$ and $Find\_Pattern2$.}   
\label{tab:FP1-FP2}
 \begin{tabular}{|c|c|c|}
 \hline
  Example & \multicolumn{2}{c|}{Computation time} \\
  \cline{2-3} & $Find\_Pattern$ & $Find\_Pattern2$   \\
  \hline 
  DC-DC Converter &    $1609$ s  &  $< 1$ s \\
  Polynomial example & Time Out   & $150$ s  \\
  Building ventilation & $272$ s & $228$ s  \\
  \hline
    
   \end{tabular}
\end{table}

Table \ref{tab:FP1-FP2} shows a comparison of functions $Find\_Pattern$ and $Find\_Pattern2$,
which shows that the new version highly improves the computation times.
We can note that the new version is all the more efficient as the length of the patterns 
increases, and as obstacles cut the research tree of patterns. This is why we 
observe significant improvements on the examples of the DC-DC converter and the polynomial example,
and not on the building ventilation example, which only requires patterns of length $2$, and presents
no obstacle.

\begin{table}
\caption{Comparison with state-of-the-art tools.}   
\label{tab:SOTA}
 \begin{tabular}{|c|c|c|c|}
 \hline
  Example & \multicolumn{3}{c|}{Computation time} \\
  \cline{2-4} & FP2 & SCOTS & PESSOA   \\ 
  \hline 
  DC-DC Converter & $< 1$ s &    $43$ s  &  $760$ s \\
  Polynomial example & $150$ s & $131$ s & $\_\_$  \\
  Unicyle \cite{zamani2012symbolic,reissig2015feedback} & $3619$ s & $492$ s & $516$ s \\
  \hline
   \end{tabular}
\end{table}

Table \ref{tab:SOTA} shows of comparison of function $Find\_Pattern2$
with state-of-the-art tools SCOTS and PESSOA.
On the example of the DC-DC converter, our algorithm manages to control
the whole state-space $R=\lbrack 1.55 , 2.15
\rbrack \times \lbrack 1.0 , 1.4 \rbrack$ in less than one second, while SCOTS and PESSOA only 
control a part of $R$, and with greater computation times. Note that these computation
times vary with the number of discretization points used in both, but even with a very
fine discretization, we never managed to control the whole box $R$.
For the polynomial example, we manage to control the whole boxes $R_1$ and $R_2$, such as SCOTS 
and in a comparable amount of time. However, PESSOA does not support natively 
this kind of nonlinear systems.
We compared our method on a last case study on which PESSOA and SCOTS perform well
(see \cite{zamani2012symbolic,reissig2015feedback} for details of this case study,
and see Appendix for a simulation obtained using our method).
For this case study, we have not obtained as good computations times as they have.
This comes from the fact that this example requires a high number of switched modes,
long patterns, as well as a high number of boxes to tile the state-space.
Note that for this case study we used an automated pre-tiling of the state-space permitting 
to decompose the reachability problem in a sequence of reachability problems. 
This is in fact the most difficult case of application of our method. 
This reveals that our method is more adapted when either the number of switched modes
of the length of patterns is not high (though it can be handled at the cost of high
computation times). 
Another advantage is that we do not require 
a homogeneous discretization of the state space. We can thus tile large parts 
of the state-space using only few boxes, and this often permits to consider much
less symbolic states than with discretization methods, especially in high dimensions (see \cite{LeCoent2016}).


\section{Conclusion}
\label{sec:conclu}

We presented a method of control synthesis for nonlinear switched
systems, based on a simple state-space bisection algorithm, and on
validated simulation. The approach permits to deal with stability, reachability,
safety and forbidden region constraints. Varying parameters and
perturbations can be easily taken into account with interval
analysis. The approach has been numerically validated on several
examples taken from the literature, a linear one with constant
parameters, and two nonlinear ones with varying perturbations.
Our approach compares well with the state-of-the art tools SCOTS and PESSOA.

We would like to point out that the exponential complexity of the algorithms presented here, 
which is inherent to guaranteed methods,
is not prohibitive. Two approaches have indeed been developed to overcome this exponential
complexity. 
A first approach is the use of compositionality, which permits to split the system in
two (or more) sub-systems, and to perform control synthesis on these sub-systems of lower dimensions.
This approach has been successfully applied in \cite{LeCoent2016} to a system of dimension $11$,
and we are currently working on 
applying this approach to the more general context of contract-based design \cite{sangiovanni2012taming}.
A second approach is the use of Model Order Reduction, which allows to 
approximate the full-order system \eqref{eq:sys} with a reduced-order system, of
lower dimension, on which it is possible to perform control synthesis.
The bounding of the trajectory errors between the full-order and the reduced-order systems can 
be taken into account, so that the induced controller is guaranteed. This approach, described in
\cite{le2016control}, has been successfully applied on (space-discretized) 
partial differential equations, leading to systems of ODEs of dimension up to $100 000$.
The present work is a potential ground for the application of such methods to control of
nonlinear partial differential equations, with the use of proper nonlinear model order reduction
techniques.

\begin{ack}                               
This work is supported by Institut Farman (project
{\scshape SWITCH\-DESIGN}), by the French National Research Agency
through the ``iCODE Institute project'' funded by the IDEX
Paris-Saclay, ANR-11-IDEX-0003-02, and by Labex
DigiCosme (project ANR-11-LABEX-0045-DIGICOSME).  
\end{ack}

\bibliographystyle{plain}        
\bibliography{minimator}           

 \vspace{8em}
\section{Appendix}
\begin{figure}[h!]
 \includegraphics[width=0.95\textwidth]{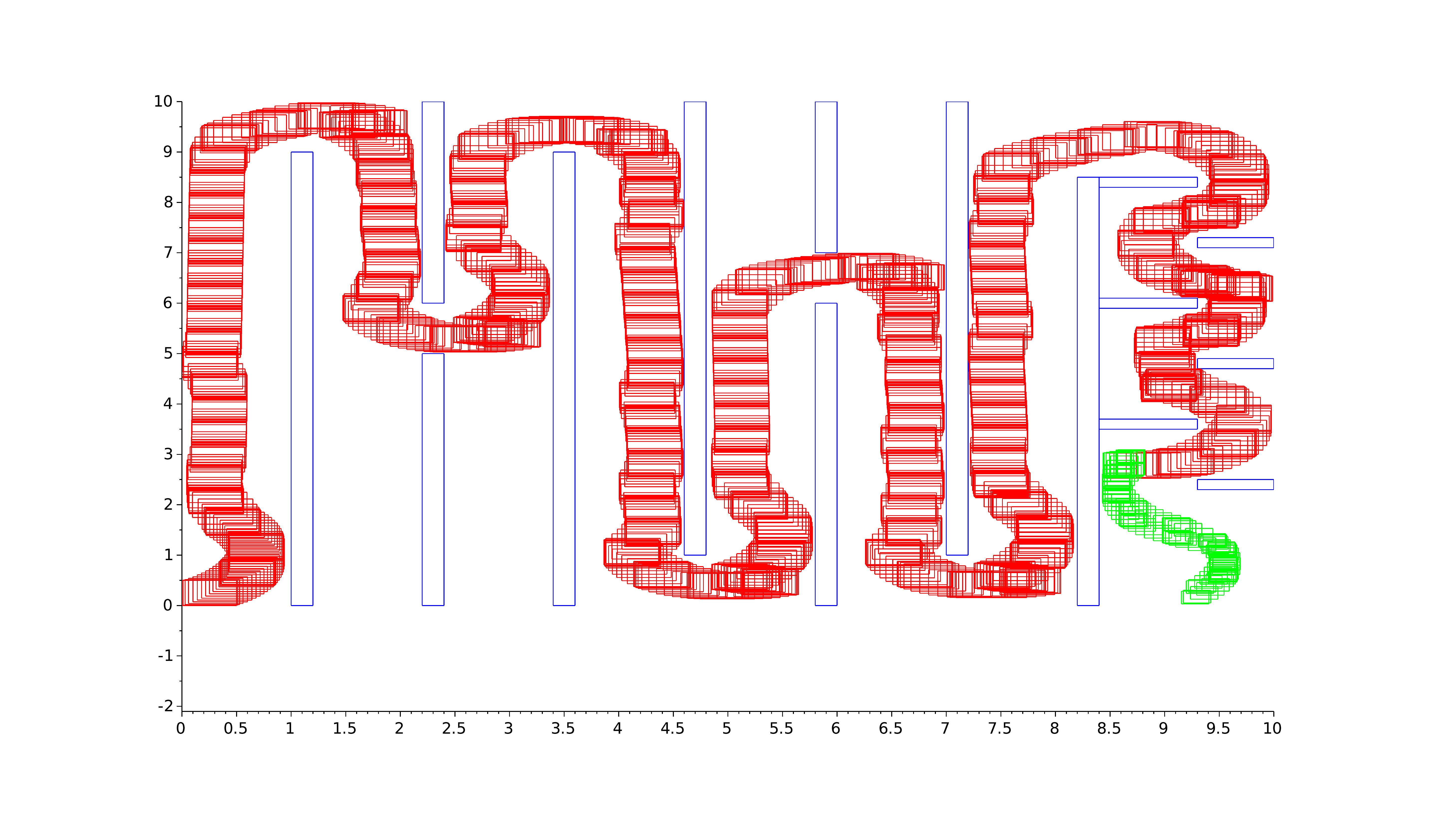}
 \caption{Simulation of the unicycle example.}
 \label{fig:unicycle}
\end{figure}


\end{document}